\documentclass[runningheads,a4paper]{llncs}
\usepackage{amssymb}
\usepackage{graphicx}
\usepackage{color,xcolor}
\usepackage{amsmath,mathrsfs}

\newcommand{\GF}[2][2]{{\mathbb F}_{{#1}^{#2}}}
\newcommand{\Tr}[2][1]{Tr_{#1}^{#2}}
\newcommand{\GFc}[1][2]{\overline{\GF{}}}

\let\BEGINPROOF\proof
\let\ENDPROOF\endproof

\renewenvironment{proof}{\BEGINPROOF}{\qed\ENDPROOF}

\title{Solving $x^{2^k+1}+x+a=0$ in $\GF{n}$ with $\gcd(n,k)=1$}

\author{Kwang Ho Kim\inst{1}\and Sihem Mesnager\inst{2}}

\institute{ Institute of Mathematics, State Academy of Sciences and
PGItech Corp.,
Pyongyang, Democratic People's Republic of Korea\\
\email{khk.cryptech@gmail.com} \and LAGA, Department of Mathematics, Universities of Paris VIII and Paris XIII, CNRS, UMR 7539 and Telecom ParisTech, France\\
\email{smesnager@univ-paris8.fr}\\}

\begin{document}

\maketitle

\makebox[\linewidth]{\today}

\begin{abstract}
Let $N_a$ be the number of solutions to the equation
$x^{2^k+1}+x+a=0$ in $\GF {n}$ where $\gcd(k,n)=1$. In 2004, by
Bluher \cite{BLUHER2004} it was known that possible values of $N_a$
are only 0, 1 and 3. In 2008, Helleseth and Kholosha
\cite{HELLESETH2008} have got criteria for $N_a=1$ and an explicit
expression of the unique solution when $\gcd(k,n)=1$. In 2014,
Bracken, Tan and Tan \cite{BRACKEN2014} presented a criterion for
$N_a=0$ when $n$ is even and $\gcd(k,n)=1$.

This paper completely solves this equation $x^{2^k+1}+x+a=0$ with
only condition $\gcd(n,k)=1$. We explicitly calculate all possible
zeros in $\GF{n}$ of $P_a(x)$. New criterion for which $a$, $N_a$ is
equal to $0$, $1$ or $3$ is a by-product of our result.
\end{abstract}


\section{Introduction}

Let $n$ be a positive integer and $\GF n$ be the finite field of
order $2^n$. The zeros of the polynomial
\begin{equation}
  \label{eq:1}
  P_a(x) = x^{2^k+1}+x+a,\quad a\in\GF n^\star
\end{equation}
has been studied in \cite{BLUHER2004,HELLESETH2008,HELLESETH2010}.
 This polynomial has
arisen in several different contexts including the inverse Galois
problem \cite{AbhyankarCohenZieve2000}, the construction of
difference sets with Singer parameters \cite{DILLON2004}, to find
cross-correlation between $m$-sequences
\cite{DOBBERTIN2006,HELLESETH2008} and to construct error correcting
codes \cite{BRACKEN2009}.  More general polynomial forms
$x^{2^k+1}+rx^{2^k}+sx+t$ are also transformed into this form by a
simple substitution of variable $x$ with $(r+s^{\frac 1{2^k}})x+r$.

It is clear that $P_a(x)$ have no multiple roots. In 2004, Bluher
\cite{BLUHER2004} proved following result.
\begin{theorem}
  \label{thm:number_zeros}
  For any $a\in\GF n^*$ and a positive integer $k$, the polynomial
  $P_a(x)$ has either none, one, three or $2^{\gcd(k,n)}+1$ zeros in
  $\GF n$.
\end{theorem}

In this paper, we will consider a particular case with
$\gcd(n,k)=1$. In this case, Theorem~\ref{thm:number_zeros} says
that $P_a(x)$ has none, one or three zeros in $\GF n$.

In 2008, Helleseth and Kholosha \cite{HELLESETH2008} have provided
criteria for which $a$ $P_a(x)$ has exactly one zero  in $\GF{n}$
 and an explicit expression of the unique zero when
$\gcd(k,n)=1$.

In 2014, Bracken, Tan and Tan \cite{BRACKEN2014} presented a
criterion for which $a$ $P_a(x)$ has no zero in $\GF{n}$ when $n$ is
even and $\gcd(k,n)=1$.

 In this paper, we explicitly calculate all possible zeros in
$\GF n$ of $P_a(x)$ when $\gcd(n,k)=1$. New criterion for which $a$,
$P_a(x)$ has none, one or three zeros is a by-product of this
result.

We begin with showing that we can reduce the study to the case when
$k$ is odd. In the odd $k$ case, one core of our approach is to
exploit a recent polynomial identity special to characteristic 2,
presented in \cite{BLUHER2016} (Theorem~\ref{thm:identity_dickson}).
This polynomial identity enables us to divide the problem of finding
zeros in $\GF n$ of $P_a$ into two independent problems: Problem 1
to find the unique preimage of an element in $\GF n$ under a
M$\ddot{u}$ller-Cohen-Matthews (MCM) polynomial
 and Problem 2 to find preimages of an
element in $\GF n$ under a Dickson polynomial
(subsection~\ref{sec:preliminaries}). There are two key stages to
solve Problem~\ref{MCMProblem}. One is to establish a relation of
the MCM polynomial with the Dobbertin polynomial. Other is to find
an explicit solution formula for the affine equation $x^{2^k}+x=b,
b\in \GF{n}$. These are done  in subsection \ref{sec:probl-refmcmpr}
and Problem~\ref{MCMProblem} is solved by
Theorem~\ref{thm:SolvingA}.
 Problem~\ref{DicksonProblem} is relatively easy which is answered by
Theorem~\ref{even_n_odd_k} and Theorem~\ref{odd_n_odd_k} in
subsection \ref{sec:probl-refd}.  Finally, we collect together all
these results to give  explicit expression of all possible zeros of
$P_a$ in $\GF n$ by Theorem~\ref{thm:maineven},
Theorem~\ref{thm:mainoddodd} and Theorem~\ref{thm:mainoddeven}.

\section{Preliminaries}


In this section, we state some results on finite fields and
introduce classical polynomials that we shall need in the sequel. We
begin with the following result that will play an important role in
our study.
\begin{proposition}
  \label{prop:decomposition}
  Let $n$ be a positive integer. Then, every element $z$ of
  $\GF n^*:=\GF{n}\setminus\{0\}$ can be written (twice) $z=c+\frac 1c$ where
  $c\in\GF n^\star:=\GF{n}\setminus \GF{}$ if $\Tr{n}(\frac{1}{z})=0$ and $c\in \mu_{2^n+1}^{\star}:=\{\zeta\in\GF {2n}\mid
\zeta^{2^n+1}=1\}\setminus \{1\}$ if $\Tr{n}(\frac{1}{z})=1$.
\end{proposition}
\begin{proof}
 For  $z\in \GF n^*$, $z=c+\frac 1c$ is equivalent to
$\frac{1}{z^2}=\frac{c}{z}+\left(\frac{c}{z}\right)^2$, and thus
this equation has a solution in $\GF{n}$ if and only if
$\Tr{n}(\frac{1}{z})=0$. Hence, mapping $c\longmapsto c+\frac 1c$ is
2-to-1 from  $\GF n$ onto $\{z\in \GF{n}\mid
\Tr{n}(\frac{1}{z})=0\}$ with convention $\frac{1}{0}:=0$. Also,
since $\left(c+\frac 1c\right)^{2^n}=c^{2^n}+\left(\frac
1c\right)^{2^n}=\frac 1c+c$ for $c\in \mu_{2^n+1}^{\star}$, the
mapping $c\longmapsto c+\frac 1c$ is 2-to-1 from
$\mu_{2^n+1}^{\star}$ with cardinality $2^n$ onto $\{z\in \GF{n}\mid
\Tr{n}(\frac{1}{z})=1\}$ with cardinality $2^{n-1}$.
\end{proof}
We shall also need two classical families of polynomials, Dickson
polynomials of the first kind and M$\ddot{u}$ller-Cohen-Matthews
polynomials.

The Dickson polynomial of the first kind of degree $k$ in
indeterminate $x$ and with parameter $a\in\GF n^*$ is
\begin{displaymath}
  D_k(x,a) = \sum_{i=0}^{\lfloor k/2\rfloor}\frac{k}{k-i}\binom{k-i}{i}a^kx^{k-2i},
\end{displaymath}
where $\lfloor k/2\rfloor$ denotes the largest integer less than or
equal to $k/2$. In this paper, we consider only Dickson polynomials
of the first kind $D_k(x,1)$, that we shall denote $D_k(x)$
throughout the paper. A classical property of Dickson polynomial
that we shall use extensively is
\begin{proposition}
  For any positive integer $k$ and any $x\in\GF n$, we have
  \begin{equation}
  \label{eq:Dickson}
  D_k\left(x+\frac 1x\right) = x^k + \frac{1}{x^k}.
\end{equation}
\end{proposition}

M$\ddot{u}$ller-Cohen-Matthews polynomials are another classical
polynomials defined as follows \cite{COHEN94},
\begin{equation*}
  f_{k,d}(X) := \frac{{T_k(X^c)}^d}{X^{2^k}}
\end{equation*}
where
\begin{equation*}
  T_k(X) := \sum_{i=0}^{k-1}X^{2^i}\quad\text{and}\quad cd = 2^k+1.
\end{equation*}
A basic property for such polynomials that we shall need in this
paper is the following statement.
\begin{theorem}
  \label{thm:MCM_permutation}
  Let $k$ and $n$ be two positive integers with $\gcd(k,n)=1$.
  \begin{enumerate}
  \item If $k$ is odd, then $f_{k,2^{k}+1}$ is a permutation
    on $\GF n$.
  \item If $k$ is even, then $f_{k,2^{k}+1}$ is a $2$-to-$1$
    on $\GF n$.
  \end{enumerate}
\end{theorem}
\begin{proof}
  For odd $k$, see \cite{COHEN94}. When $k$ is even, $f_{k,2^k+1}$ is
  not a permutation of $\GF n$. Indeed, Theorem 10 of
  \cite{DILLON2004} states that $f_{k,1}$ is $2$-to-$1$, and then the statement follows
  from the facts that $f_{k,2^k+1}(x^{2^k+1})=f_{k,1}(x)^{2^k+1}$ and
  $\gcd(2^k+1, 2^n-1)=1$ when $\gcd(k,n)=1$.
\end{proof}

We will exploit a recent polynomial identity involving Dickson
polynomials established in \cite[Theorem 2.2]{BLUHER2016}.
\begin{theorem}
  \label{thm:identity_dickson}
  In the polynomial ring $\GF[2^k]{}[X,Y]$, we have the identity
  \begin{equation*}
    X^{2^{2k}-1}+\left(\sum_{i=1}^kY^{2^k-2^i}\right)X^{2^k-1}+Y^{2^k-1} = \prod_{w\in\GF[2^k]{}^{*}} \left(D_{2^k+1}(wX)
    -Y\right).
  \end{equation*}
\end{theorem}

Finally we remark that the identity  by Abhyankar, Cohen, and Zieve
\cite[Theorem 1.1]{AbhyankarCohenZieve2000} tantalizingly similar to
this identity treats any characteristic, while this identity is
special to characteristic 2 (this may happen because the Dickson
polynomials are ramified at the prime 2). However, the
Abhyankar-Cohen-Zieve identity has not lead us to solving
$P_a(x)=0$.
\section{Solving $P_a(x)=0$}
Throughout this section, $k$ and $n$ are coprime and we set
$q=2^k$.
\subsection{Splitting the problem}
\label{sec:preliminaries}

One core of our approach is to exploit Theorem
\ref{thm:identity_dickson} to the study of zeros in $\GF n$ of
$P_a$. To this end, we observe firstly that
\begin{equation*}
  \begin{split}
    &X^{q^2-1}+\left(\sum_{i=1}^kY^{q-2^i}\right)X^{q-1}+Y^{q-1}\\
    &\qquad = X^{q^2-1} + Y^{q} T_k\left(\frac
    1Y\right)^2\,X^{q-1}+Y^{q-1}.
  \end{split}
\end{equation*}
Substituting $tx$ to $X$ in the above identity with
$t^{q^2-q}=Y^{q}T_k\left(\frac 1Y\right)^2$, we get
\begin{equation*}
  \begin{split}
    &X^{q^2-1}+\left(\sum_{i=1}^kY^{q-2^i}\right)X^{q-1}+Y^{q-1}\\
    &\qquad= Y^{q}T_k\left(\frac 1Y\right)^2t^{q-1}\left(x^{q^2-1}+x^{q-1} +
      \frac{1}{YT_k\left(\frac 1Y\right)^{2}t^{q-1}}\right).\\
  \end{split}
\end{equation*}
Now, $t^{q^2-q}=Y^{q}T_k\left(\frac 1Y\right)^2$ is equivalent to
$t^{q-1}=YT_k\left(\frac 1Y\right)^{\frac 2q}$. Therefore
\begin{equation*}
  \begin{split}
    YT_k\left(\frac 1Y\right)^{2}t^{q-1} = Y^{2}T_k\left(\frac 1Y\right)^{\frac{2(q+1)}q} = \left(f_{k,q+1}\left(\frac 1Y\right)\right)^{\frac 2q}.
  \end{split}
\end{equation*}
By all these calculations, we get
\begin{equation}
  \label{eq:A}
  \begin{split}
    &x^{q^2-1}+x^{q-1} + \frac{1}{ \left(f_{k,q+1}\left(\frac 1Y\right)\right)^{\frac 2q}} \\
    &\qquad= \frac{1}{Y^{q-1} \left(f_{k,q+1}\left(\frac 1Y\right)\right)^{\frac 2q}}\left(X^{q^2-1}+\left(\sum_{i=1}^kY^{q-2^i}\right)X^{q-1}+Y^{q-1}\right).
  \end{split}
\end{equation}

If $k$ is odd, $f_{k,q+1}$ is a permutation polynomial of $\GF n$ by
Theorem~\ref{thm:MCM_permutation}. Therefore, for any $a\in\GF n^*$,
there exists a unique $Y$ in $\GF n^*$ such that
$a=\frac{1}{f_{k,q+1}\left(\frac 1Y\right)^{\frac 2q}}$. Hence, by
Theorem~\ref{thm:identity_dickson} and equation~(\ref{eq:A}), we
have
\begin{equation}
  \label{eq:Dickson1}
  P_a\left(x^{q-1}\right) = x^{q^2-1}+x^{q-1}+a = \frac{1}{Y^{q-1} \left(f_{k,q+1}\left(\frac 1Y\right)\right)^{\frac 2q}}\prod_{w\in\GF[q]{}^*}\left(D_{q+1}\left(wtx\right)-Y\right)
\end{equation}
where $Y$ is the unique element of $\GF n^*$ such that
$a=\frac{1}{f_{k,q+1}\left(\frac 1Y\right)^{\frac 2q}}$ and
$t^{q-1}=YT_k\left(\frac 1Y\right)^{\frac 2q}$. Now, since
$\gcd(q-1,2^n-1)=1$, the zeros of $P_a(x)$ are the images of the
zeros of $P_a(x^{q-1})$ by the map $x\mapsto x^{q-1}$. Therefore,
when $k$ is odd, equation~(\ref{eq:Dickson1}) states that finding
the zeros of $P_a(x^{q-1})$ amounts to determine preimages of $Y$
under the Dickson polynomial $D_{q+1}$.

When $k$ is even, $f_{k,q+1}$ is no longer a permutation and we
cannot repeat again the preceding argument (indeed, when $k$ is
even, $f_{k,q+1}$ is $2$-to-$1$, see
Theorem~\ref{thm:MCM_permutation}). Fortunately, we can go back to
the odd case by rewriting the equation. Indeed, for $x\in \GF{n}$,
\begin{equation*}
  \begin{split}
    P_a(x) &= x^{2^k+1}+x+a = \left(x^{2^{n-k}+1}+x^{2^{n-k}}+a^{2^{n-k}}\right)^{2^k}\\
    &=  \left((x+1)^{2^{n-k}+1}+(x+1)+a^{2^{n-k}}\right)^{2^k}
  \end{split}
\end{equation*}  and so
\begin{equation}
    \label{eq:even_case}
    \{x\in\GF n\mid P_a(x) = 0\} = \left\{x + 1\mid x^{2^{n-k}+1}+x+a^{2^{n-k}}=0,\,x\in\GF n\right\}.
  \end{equation}
If $k$ is even, then $n$ is odd as $\gcd(k,n)=1$, and so $n-k$ is
odd and we can reduce to the odd case.

We now summarize all the above discussions in the following theorem.
\begin{theorem}
  \label{thm:main}
  Let $k$ and $n$ be two positive integers such that $\gcd(k,n)=1$.
  \begin{enumerate}
  \item\label{oddcase} Let $k$ be odd and $q=2^k$. Let
    $Y\in\GF n^*$ be (uniquely) defined by
    $a=\frac{1}{f_{k,q+1}\left(\frac 1Y\right)^{\frac 2q}}$. Then,
    \begin{equation*}
      \{x\in\GF n\mid P_a(x) = 0\} = \left\{\frac{z^{q-1}}{YT_k\left(\frac 1Y\right)^{\frac 2q}}\,|\, D_{q+1}(z) = Y,\, z\in\GF n\right\}.
    \end{equation*}
  \item Let $k$ be even and $q^\prime=2^{n-k}$. Let $Y^\prime\in\GF n^*$ be (uniquely) defined by
    $a^{q^\prime}=\frac{1}{f_{n-k,q^\prime+1}\left(\frac
        1{Y^\prime}\right)^{\frac 2{q^\prime}}}$. Then,
    \begin{equation*}
      \{x\in\GF n\mid P_a(x) = 0\} = \left\{1+\frac{z^{q^\prime-1}}{Y^\prime T_{n-k}\left(\frac 1{Y^\prime}\right)^{\frac 2{q^\prime}}}\,|\, D_{q^\prime+1}(z) = Y^\prime,\, z\in\GF n\right\}.
    \end{equation*}
  \end{enumerate}
\end{theorem}
\begin{proof}
  Suppose that $k$ is odd.
  Equation~(\ref{eq:Dickson1}) shows that the zeros of $P_a$ in $\GF n$
  are $x^{q-1}$ for the elements $x\in\GF n^\star$ such that $D_{q+1}(wtx)=Y$ where $t^{q-1}=YT_k\left(\frac 1Y\right)^{\frac 2q}$.
  Set $z=wtx$. Then, since
  $w\in\GF[q]{}^*$, $x^{q-1}=\left(\frac{z}{wt}\right)^{q-1}=\frac{z^{q-1}}{t^{q-1}}=\frac{z^{q-1}}{YT_k\left(\frac 1Y\right)^{\frac 2q}}$.
   Item 2 follows from Item \ref{oddcase} and equality
   \eqref{eq:even_case}.
\end{proof}
Theorem~\ref{thm:main} shows that we can split the problem of
finding the zeros in $\GF n$ of $P_a$ into two independent problems
with odd $k$.
\begin{problem}
  \label{MCMProblem} For $a\in\GF n^*$, find the unique element
  $Y$ in $\GF n^*$ such that
  \begin{equation}\label{eq:ProblemA}
    a^{\frac q2} = \frac{1}{f_{k,q+1}\left(\frac 1Y\right)}.
  \end{equation}
\end{problem}
\begin{problem}
  \label{DicksonProblem}
  For $Y\in\GF n^*$, find the preimages in $\GF n$ of $Y$ under
  the Dickson polynomial $D_{q+1}$, that is, find the elements of the
  set
  \begin{equation}\label{eq:ProblemB}
  D_{q+1}^{-1}(Y) = \{z\in\GF n^\star\mid D_{q+1}(z)= Y\}.
  \end{equation}
\end{problem}
In the following two subsections, we shall study those two problems
only when $k$ is odd since, if $k$ is even, it suffices to replace
$k$ by $n-k$, $q$ by $q^\prime=2^{n-k}$, and $a$ by $a^{q^\prime}$
in all the results of the odd case.

\subsection{On problem~\ref{MCMProblem}}
\label{sec:probl-refmcmpr}

Define
\begin{equation}
  \label{eq:Qkk'}
  Q_{k,k^\prime}^\prime(x) =\frac{x^{q+1}}{\sum_{i=1}^{k^\prime} x^{q^i}}
\end{equation}
where $k^\prime < 2n$ is the inverse of $k$ modulo $2n$, that is,
s.t. $kk^\prime=1\mod 2n$. Note that $k^\prime$ is odd since
$\gcd(k^\prime, 2n)=1$. It is known that if $\gcd(2n,k)=1$ and
$k^\prime$ is odd, then $Q_{k,k^\prime}^\prime$ is permutation on
$\GF {2n}$ (see \cite{DILLON2004} or \cite{DILLON99} where
$Q_{k,k^\prime}=1/Q_{k,k^\prime}^\prime$ is instead considered).
Indeed, due to \cite{DILLON2004}, defining the following sequences
of polynomials
  \[
    A_1(x)=x,\,
    A_2(x)=x^{q+1},\,
    A_{i+2}(x)=x^{q^{i+1}} A_{i+1}(x)+x^{q^{i+1}-q^i}A_i(x), \quad i\geq 1,
  \]

  \[
    B_1(x)=0,\,
    B_2(x)=x^{q-1},\,
    B_{i+2}(x)=x^{q^{i+1}} B_{i+1}(x)+x^{q^{i+1}-q^i}B_i(x), \quad i\geq
    1,
  \]
  then the polynomial expression of the inverse $R_{k,k^\prime}$ of the mapping induced by $Q_{k,k^\prime}^\prime$ on $\GF{2n}$ is
\begin{equation}
  \label{eq:Rkk'}
  R_{k,k'}(x)=\sum_{i=1}^{k'}A_i(x)+B_{k'}(x).
\end{equation}

 Directively from the definitions, it follow
 \[f_{k,q+1}(x+x^2) = \frac{(x+x^q)^{q+1}}{x^q+x^{2q}}\]
and
 \[ Q_{k,k^\prime}^\prime\left(x+x^{q}\right)= \frac{(x+x^q)^{q+1}}{x^q+x^{q^{k^\prime+1}}}.\]
Since $x^{2q}=x^{q^{k^\prime+1}}\Longleftrightarrow
x=x^{2^{kk^\prime-1}}$, it holds that
\begin{equation}
  \label{eq:DillonDobbertin}
  f_{k,q+1}(x+x^2) =  Q_{k,k^\prime}^\prime\left(x+x^{q}\right)
\end{equation}
for any $x\in \GF{2n}$. Let $x$ be an element of $\GF{2n}$ such that
\[\frac{1}{Y}=x+x^2.\]  By using
\eqref{eq:DillonDobbertin} we can rewrite \eqref{eq:ProblemA} as
 \[ a^{-\frac q2} =Q_{k,k^\prime}^\prime\left(x+x^{q}\right).\]

 Therefore, we
have
\begin{proposition}
  \label{thm:MCMLinearEquation}
  Let $a\in\GF n^*$. Let
  $x\in\GF{2n}$ be a solution of
\begin{displaymath}
   R_{k,k^\prime}\left(a^{-\frac q2}\right) =x+x^{q}.  \end{displaymath}
 Then,
  $Y=\frac{1}{x+x^{2}}=\left(1+\frac{1}{x}\right)+\frac{1}{\left(1+\frac{1}{x}\right)}$ is
  the unique solution in $\GF{n}$ of
  $a^{\frac q2} = \left(f_{k,q+1}\left(\frac 1Y\right)\right)^{-1}$.
\end{proposition}
Proposition~\ref{thm:MCMLinearEquation} shows that solving
Problem~\ref{MCMProblem} amounts to find a solution of an affine
equation $x+x^q=b$, which is done in the following.
\begin{proposition}
  \label{MCM:Solvingx^q+x=b}
  Let $k$ be odd and $\gcd(n,k)=1$.  Then, for any
  $b\in\GF n$,
  \begin{displaymath}
    \{x\in\GF {2n}\mid
    x+x^q=b\}=S_{n,k}\left(\frac{b}{\zeta+1}\right)+\GF{},
  \end{displaymath}
  where $S_{n,k}(x)=\sum_{i=0}^{n-1}x^{q^{i}}$ and $\zeta$ is an element of $\mu_{2^n+1}^{\star}$.
\end{proposition}
\begin{proof} As it was assumed that $k$ is odd and $\gcd(n,k)=1$,
it holds $\gcd(2n,k)=1$ and so the linear mapping
$x\in\GF{2n}\longmapsto x+x^q$ has kernel of dimension 1, i.e. the
equation $x+x^q=b$ has at most 2 solutions in $\GF{2n}$.
  Since
  $S_{n,k}(x) + \left(S_{n,k}(x)\right)^{q} = x + x^{q^{n}}$, we
  have
  \begin{eqnarray*}
    S_{n,k}\left(\frac{b}{\zeta+1}\right)+\left(S_{n,k}\left(\frac{b}{\zeta+1}\right)\right)^{q} + b
    &=& \frac{b}{\zeta+1} + \left(\frac{b}{\zeta+1}\right)^{q^{n}} + b\\
    &=& \frac{b}{\zeta+1} + \frac{b}{\zeta^{q^{n}}+1} + b\\
    &=&  \frac{b}{\zeta+1} + \frac{b}{1/\zeta+1} + b\\
    &=& 0
  \end{eqnarray*}
   and thus really $S_{n,k}\left(\frac{b}{\zeta+1}\right), S_{n,k}\left(\frac{b}{\zeta+1}\right)+1\in \GF{2n}$
   are   the $\GF{2n}-$solutions of the equation $x+x^q=b$.
\end{proof}
By Proposition~\ref{thm:MCMLinearEquation} and
Proposition~\ref{MCM:Solvingx^q+x=b}, we can now explicit the
solutions of Problem~\ref{MCMProblem}.
\begin{theorem}
  \label{thm:SolvingA}  Let $a\in\GF n^*$. Let $k$ be odd with $\gcd(n,k)=1$ and $k^\prime$ be the
  inverse of $k$ modulo $2n$. Then, the unique solution of
  (\ref{eq:ProblemA}) in $\GF n^*$ is
    \begin{displaymath}
    Y = \frac{1}{S_{n,k}\left(\frac{R_{k,k^\prime}\left(a^{-\frac q2}\right)}{\zeta+1}\right)+\left(S_{n,k}\left(\frac{R_{k,k^\prime}\left(a^{-\frac q2}\right)}{\zeta+1}\right)\right)^2}
  \end{displaymath}
  where $\zeta$ denotes any element of $\GF{2n}^\star$ such that
    $\zeta^{2^n+1}=1$, $S_{n,k}(x)=\sum_{i=0}^{n-1}x^{q^{i}}$ and
    $R_{k,k^\prime}$ is defined by (\ref{eq:Rkk'}). Furthermore, we
    have $Y=T+\frac{1}{T}$ for
    \begin{displaymath}
    T=1+\frac{1}{S_{n,k}\left(\frac{R_{k,k^\prime}\left(a^{-\frac
    q2}\right)}{\zeta+1}\right)}.
  \end{displaymath}

\end{theorem}

\subsection{On Problem~\ref{DicksonProblem}}
\label{sec:probl-refd}
By Proposition~\ref{prop:decomposition}, one
can write $z=c+\frac{1}{c}$ where $c\in\GF{n}^\star$ or $c\in
\mu_{2^n+1}^{\star}$. Equation~(\ref{eq:Dickson}) applied to $z$
leads then to
\begin{equation}
  \label{eq:ADickson}
  D_{q+1}(z) = c^{q+1} + \frac{1}{c^{q+1}}.
\end{equation}
Thus, we can be reduced to solve firstly equation $T+\frac{1}{T}=Y$,
then equation $c^{q+1}=T$ in $\GF{n}^\star \cup
\mu_{2^n+1}^{\star}$, and set $z=c+\frac{1}{c}$. Here, let us point
out that $c^{q+1}=T$ is equivalent to $\left(\frac
1c\right)^{q+1}=\frac 1T$ and that $c$ and $\frac 1c$ define the
same element $z=c+\frac 1c$ of $\GF n$.

 Proposition~\ref{prop:decomposition} says that the equation $T+\frac 1T = Y$ has two solutions in $\GF
 n^\star$ if $\Tr{n}\left(\frac 1Y\right)=0$ and in $\mu_{2^n+1}^{\star}$ if $\Tr{n}\left(\frac
 1Y\right)=1$. In fact, Proposition~\ref{MCM:Solvingx^q+x=b} gives
 an explicit solution expression, that is,
 \begin{equation}\label{solution:T+1/T=Y}
 T=YS_{n,1}\left(\frac{1}{Y^2(\zeta+1)}\right)
 \text{ and }
 T=YS_{n,1}\left(\frac{1}{Y^2(\zeta+1)}\right)+Y,
 \end{equation}
 where
 $S_{n,1}(x)=\sum_{i=0}^{n-1}x^{2^i}$ and $\zeta$ is any element of
 $\mu_{2^n+1}^{\star}$.

Now, let us consider solutions of $c^{q+1}=T$ in $\GF{n}^\star \cup
\mu_{2^n+1}^{\star}$. First, note that if $T\in \GF{n}^\star$, then
necessarily $c\in\GF n^\star$ (indeed, if $c\in
\mu_{2^n+1}^{\star}$, we get $T^2=T\cdot T=T^{2^n}\cdot
T=T^{2^n+1}=(c^{2^n+1})^{q+1}=1$ contradicting $T\notin \GF{}$).

Recall that if $k$ is odd and $\gcd(n,k)=1$, then
\begin{equation}\label{gcd:+-}
  \gcd(q+1,2^n-1) = \begin{cases}1, & \mbox{if $n$ is odd}\\ 3, & \mbox{if $n$ is
  even}\end{cases}
\end{equation}
and
\begin{equation}\label{gcd:++}
  \gcd(q+1,2^n+1) = \begin{cases}1, & \mbox{if $n$ is even}\\ 3, & \mbox{if $n$ is odd.}\end{cases}
\end{equation}

Therefore, if $T\in \GF{n}^\star$, then there are $0$ (if $T$ is a
non-cube in $\GF{n}^\star$) or $3$ (if $T$ is a cube in
$\GF{n}^\star$) elements $c$ in $\GF n^\star$ such that $c^{q+1}=T$
when $n$ is even while there is a unique $c$ (i.e.
$T^{(q+1)^{-1}\mod 2^n-1}$) when $n$ is odd.  And, if $T\in
\mu_{2^n+1}^{\star}$, then there are $0$ (if $T$ is a non-cube in
$\mu_{2^n+1}^{\star}$) or $3$ (if $T$ is a cube in
$\mu_{2^n+1}^{\star}$) elements $c$ in $\mu_{2^n+1}^{\star}$ such
that $c^{q+1}=T$ when $n$ is odd while there is a unique $c$ (i.e.
$T^{(q+1)^{-1}\mod 2^n+1}$) when $n$ is even.

It remains to show in the case when there are three solutions $c$,
they define three different elements $z\in \GF n^\star$. Denote $w$
a primitive element of $\GF[4]{}$. Then these three solutions of
$c^{q+1}=T$ are of form $c$, $cw$ and $cw^2$. Now,
$cw_1+\frac{1}{cw_1}=cw_2+\frac{1}{cw_2}$ implies that $cw_1=cw_2$
or $cw_1=\frac{1}{cw_2}$ (because $A+\frac 1A = B + \frac 1B$ is
equivalent to $(A+B)(AB + 1) = 0$). The second case is impossible
because it implies that
$T=c^{q+1}=\left(\frac{1}{w_1^{\frac12}w_2^{\frac12}}\right)^{q+1}=1$
because $3$ divides $q+1$ when $k$ is odd.

We can thus state the following answer to Problem 2.
\begin{theorem}
  \label{even_n_odd_k}
  Let $k$ be odd and $n$ be even. Let $Y\in\GF n^*$. Let $T$ be any element of $\GF {2n}$ such that $T+\frac 1T=Y$ (this can be given by \eqref{solution:T+1/T=Y}).
  \begin{enumerate}
  \item If $T$ is a non-cube in $\GF{n}^\star$, then
   \[D_{q+1}^{-1}(Y)=\emptyset.\]
  \item If $T$ is a cube in $\GF{n}^\star$, then
  \begin{equation*}
     D_{q+1}^{-1}(Y)=\left\{cw+\frac{1}{cw}\mid c^{q+1} = T,\,c\in\GF n^\star,\,
     w\in\GF[4]{}^*\right\}.
   \end{equation*}
  \item If $T$ is not in $\GF{n}$,  then
   \begin{equation*}
     D_{q+1}^{-1}(Y)=\left\{T^{(q+1)^{-1}\mod 2^n+1}+\frac{1}{T^{(q+1)^{-1}\mod 2^n+1}}\right\}.
   \end{equation*}
 \end{enumerate}
\end{theorem}

\begin{remark}
  Item 1 of Theorem~\ref{even_n_odd_k} recovers \cite[Theorem
  2.1]{BRACKEN2014} which states: when $n$ is even and $\gcd(n,k)=1$ (so $k$ is odd), $P_a$ has no zeros in $\GF n$ if
  and only if $a^{-1}=f_{k,q+1}\left(\frac{1}{T+\frac1T}\right)^{\frac 2q}$ for
  some non-cube $T$ of $\GF n^\star$.  Indeed, the statement of Theorem 2.1 in \cite{BRACKEN2014} is not
  exactly what we write but it is worth noticing that the quantity
  that is denoted $A(b)$ in \cite{BRACKEN2014} satisfies
  ${A(b)}^{-1}=f_{k,q+1}\left(\frac{1}{b^{\frac14}+\frac 1{b^{\frac
          14}}}\right)^{\frac2q}$.
\end{remark}
\begin{theorem}
  \label{odd_n_odd_k}
 Let $k$ be odd and $n$ be odd. Let $Y\in\GF n^*$. Let $T$ be any element of $\GF {2n}$ such that $T+\frac 1T=Y$ (this can be given by \eqref{solution:T+1/T=Y}).
  \begin{enumerate}
  \item If $T$ is a non-cube in $\mu_{2^n+1}^{\star}$, then
   \[D_{q+1}^{-1}(Y)=\emptyset.\]
  \item If $T$ is a cube in $\mu_{2^n+1}^{\star}$, then
  \begin{equation*}
     D_{q+1}^{-1}(Y)=\left\{cw+\frac{1}{cw}\mid c^{q+1} = T,\,c\in\mu_{2^n+1}^{\star},\,
     w\in\GF[4]{}^*\right\}.
   \end{equation*}
  \item If $T$ is in $\GF{n}$,  then
   \begin{equation*}
     D_{q+1}^{-1}(Y)=\left\{T^{(q+1)^{-1}\mod 2^n-1}+\frac{1}{T^{(q+1)^{-1}\mod 2^n-1}}\right\}.
   \end{equation*}
 \end{enumerate}
\end{theorem}

\subsection{On the roots in $\GF{n}$ of $P_a(x)$}
\label{sec:zeros-refeq:1}

We sum up the results of previous subsections to give an explicit
expression of the roots in $\GF{n}$ of $P_a(x)$.

Let $k$ denote any positive integer coprime with $n$ and $a\in\GF
n^*$.

First, let us consider the case of odd $k$.  Let $k^\prime$ be the
inverse of $k$ modulo $2n$.  Define
\begin{displaymath}
  T=1+\frac{1}{S_{n,k}\left(\frac{R_{k,k^\prime}\left(a^{-\frac
  q2}\right)}{\zeta+1}\right)},
\end{displaymath}
 where $\zeta$ is any element of $\GF{2n}^\star$ such that
    $\zeta^{2^n+1}=1$, $S_{n,k}(x)=\sum_{i=0}^{n-1}x^{q^{i}}$ and
    $R_{k,k^\prime}$ is defined by (\ref{eq:Rkk'}).

According to Theorem~\ref{thm:SolvingA}, Theorem~\ref{even_n_odd_k}
and Theorem~\ref{odd_n_odd_k}, we have followings.
\begin{theorem}
  \label{thm:maineven}
  Let $n$ be even, $\gcd(n,k)=1$ and $a\in \GF{n}^*$.
  \begin{enumerate}
  \item If $T$ is a non-cube in $\GF n$, then $P_a(x)$ has no zeros in $\GF n$.
  \item If $T$ is a cube in $\GF n$,
    then $P_a(x)$ has three distinct zeros
    $\frac{\left(cw+\frac 1{cw}\right)^{q-1}}{YT_k\left(\frac 1Y\right)^{\frac
        2q}}$ in $\GF n$, where $c^{q+1}=T$,
    $w\in\GF[4]{}^*$ and $Y=T+\frac 1T$.
  \item If $T$ is not in $\GF{n}$, then $P_a(x)$
    has a unique zero  $\frac{\left(c+\frac{1}{c}\right)^{q-1}}{YT_k\left(\frac 1Y\right)^{\frac
        2q}}$ in $\GF n$, where $c={T}^{(q+1)^{-1}\mod 2^n+1}$ and $Y=T+\frac 1T$.
\end{enumerate}
\end{theorem}
\begin{remark}
  When $k=1$, that is, $P_a(x)=x^3+x+a$,  Item (1) of Theorem~\ref{thm:maineven} is exactly Corollary
  2.2 of \cite{BRACKEN2014} which states that, when $n$ is even, $P_a$
  is irreducible over $\GF n$ if and only if $a=c+\frac1c$ for some
  non-cube $c$ of $\GF n$.
\end{remark}

\begin{theorem}
  \label{thm:mainoddodd}
  Let $n$ and $k$ be odds with $\gcd(n,k)=1$ and $a\in \GF{n}^*$.
  \begin{enumerate}
  \item If $T$ is a non-cube in $\mu_{2^n+1}^{\star}$,
    then $P_a(x)$ has no zeros in $\GF n$.
  \item If $T$ is a cube in $\mu_{2^n+1}^{\star}$,
    then $P_a(x)$ has three distinct zeros
    $\frac{\left(cw+\frac 1{cw}\right)^{q-1}}{YT_k\left(\frac 1Y\right)^{\frac
        2q}}$ in $\GF n$, where $c^{q+1}=T$,
    $w\in\GF[4]{}^*$ and $Y=T+\frac 1T$.
  \item If $T$ is in $\GF n$, then  $P_a(x)$ has a unique
    zero  $\frac{\left(c+\frac{1}{c}\right)^{q-1}}{YT_k\left(\frac 1Y\right)^{\frac
        2q}}$ in $\GF n$, where $c={T}^{(q+1)^{-1}\mod 2^n-1}$ and $Y=T+\frac 1T$.
\end{enumerate}
\end{theorem}
When $k$ is even, following Item (2) of Theorem~\ref{thm:main}, we
introduce $l=n-k$, $q^\prime=2^{l}$ and $l^\prime$ the inverse of
$l$ modulo $2n$.  Define
\begin{displaymath}
  T^\prime=1+\frac{1}{S_{n,l}\left(\frac{R_{l,l^\prime}\left(a^{-\frac
  {(q^\prime)^2}{2}}\right)}{\zeta+1}\right)},
\end{displaymath}
 where $\zeta$ is any element of $\GF{2n}^\star$ such that
    $\zeta^{2^n+1}=1$, $S_{n,l}(x)=\sum_{i=0}^{n-1}x^{{q^\prime}^{i}}$ and
    $R_{l,l^\prime}$ is defined by (\ref{eq:Rkk'}).

\begin{theorem}
  \label{thm:mainoddeven}
  Let $n$ be odd and $k$ be even with $\gcd(n,k)=1$. Let $a\in
  \GF{n}^*$.
  \begin{enumerate}
  \item If $T^\prime$ is a non-cube $\mu_{2^n+1}^{\star}$,
    then $P_a(x)$ has no zeros in $\GF n$.
  \item If $T^\prime$ is a cube in $\mu_{2^n+1}^{\star}$,
    then $P_a(x)$ has three distinct zeros
    $1+\frac{\left(dw+\frac 1{dw}\right)^{q^\prime-1}}{Y^\prime T_{l}\left(\frac 1{Y^\prime}\right)^{\frac
        2{q^\prime}}}$ in $\GF n$, where $d^{q^{\prime}+1}=T^\prime$,
    $w\in\GF[4]{}^*$ and $Y^\prime=T^\prime+\frac 1{T^\prime}$.
  \item If $T^\prime$ is in $\GF n$, then $P_a(x)$ has a unique
    zero
    $1+\frac{{\left(c+\frac{1}{c}\right)}^{q^\prime-1}}{Y^\prime
      T_{l}\left(\frac 1{Y^\prime}\right)^{\frac 2{q^\prime}}}$ in $\GF n$,
      where $c={T^\prime}^{(q^\prime+1)^{-1}\mod 2^n-1}$ and
    $Y^\prime=T^\prime+\frac 1{T^\prime}$.
  \end{enumerate}
\end{theorem}

\begin{remark}
When $n$ is even, Theorem \ref{thm:maineven} shows that $P_a$ has a
unique solution if and only if $T$ is not in $\GF n$. According to
Proposition \ref{thm:MCMLinearEquation}, this is equivalent to $\Tr
n(R_{k,k^\prime}(a^{-\frac q2}))=1$, that is, $\Tr
n(R_{k,k^\prime}(a^{-1}))=1$. When $n$ is odd and $k$ is odd (resp.
even), Theorem \ref{thm:mainoddodd} and Theorem
\ref{thm:mainoddeven} show that $P_a$ has a unique zero in $\GF n$
if and only if $T$ (resp. $T^\prime$) is in $\GF n$. According to
Proposition \ref{thm:MCMLinearEquation}, this is equivalent to $\Tr
n(R_{k,k^\prime}(a^{-1}))=0$ or  $\Tr n(R_{l,l^\prime}(a^{-1}))=0$
for odd $k$ or even $k$, respectively. By the way, for $x\in
\GF{n}$,
$Q_{l,l^\prime}^\prime\left(x+x^{q^\prime}\right)=\frac{(x+x^{q^\prime})^{q^\prime+1}}{x^{q^\prime}+x^{2{q^\prime}}}=\left(\frac{(x+x^q)^{q+1}}{x^q+x^{2q}}\right)^{2^{(n-k)^2}}=Q_{k,k^\prime}^\prime\left(x+x^{q}\right)^{2^{(n-k)^2}}$.
Hence if $T^\prime \in \GF{n}$, then
$R_{l,l^\prime}(a^{-1})=R_{k,k^\prime}(a^{-1})^{\frac{1}{2^{(n-k)^2}}}$,
and so $\Tr n(R_{l,l^\prime}(a^{-1}))=0$ is equivalent to $\Tr
n(R_{k,k^\prime}(a^{-1}))=0$. After all, we can recover
\cite[Theorem 1]{HELLESETH2008} which states that $P_a$ has a unique
zero in $\GF n$ if
  and only if $\Tr n(R_{k,k^\prime}(a^{-1})+1)=1$.
\end{remark}

\section{Conclusion}

In
\cite{BLUHER2004,BLUHER2016,BRACKEN2014,HELLESETH2008,HELLESETH2010},
 partial results about the zeros of $P_a(x)=x^{2^k+1}+x+a$ in $\GF
n$ have been obtained. In this paper, we provided explicit
expression of all possible roots in $\GF n$ of $P_a(x)$ in terms of
$a$ and thus finish the study initiated in these papers when
$\gcd(n,k)=1$. We showed that the problem of finding zeros in $\GF
n$ of $P_a(x)$ in fact can be divided into two problems with odd
$k$: to find the unique preimage of an element in $\GF n$ under a
M$\ddot{u}$ller-Cohen-Matthews (MCM) polynomial
 and to find preimages of an
element in $\GF n$ under a Dickson polynomial. We completely solved
these two independent problems. We also presented an explicit
solution formula for the affine equation $x^{2^k}+x=b, b\in \GF{n}$.


\begin{thebibliography}{1}

\bibitem{AbhyankarCohenZieve2000}
S.S. Abhyankar, S.D.Cohen, M.E. Zieve. Bivariate factorizations
connecting Dickson polynomials and Galois theory, Transactions of
the American Mathematical Society, 352(6):2871 -- 2887, 2000.


\bibitem{BLUHER2004}
A.W. Bluher.
\newblock On $x^{q+1}+ax+b$.
\newblock {\em Finite Fields and Their Applications}, 10(3):285 -- 305, 2004.

\bibitem{BLUHER2016}
A.W. Bluher.
\newblock A New Identity of Dickson polynomials.
\newblock eprint arXiv:1610.05853v1, October 2016.

\bibitem{BRACKEN2009}
C. Bracken, T. Helleseth.
\newblock Triple-error-correcting bch-like codes.
\newblock In {\em {IEEE} International Symposium on Information Theory, {ISIT}
  2009, June 28 - July 3, 2009, Seoul, Korea, Proceedings}, pages 1723--1725,
  2009.

\bibitem{BRACKEN2014}
C. Bracken, C.H. Tan, Y. Tan.
\newblock On a class of quadratic polynomials with no zeros and its application
  to APN functions.
\newblock {\em Finite Fields and Their Applications}, 25:26 -- 36, 2014.

\bibitem{COHEN94}
S.D. Cohen, R.W. Matthews.
\newblock A class of exceptional polynomials.
\newblock {\em Transactions of the American Mathematical Society}, 345:897 --
  909, 1994.

\bibitem{DILLON2004}
J.F. Dillon, H, Dobbertin.
\newblock New cyclic difference sets with singer parameters.
\newblock {\em Finite Fields and Their Applications}, 10(3):342 -- 389, 2004.

\bibitem{DILLON99}
J.F. Dillon.
\newblock Multiplicative difference sets via additive characters.
\newblock {\em Designs, Codes and Cryptography}, 17:225 -- 235, 1999.

\bibitem{DOBBERTIN2006}
H. Dobbertin, P. Felke, T. Helleseth, P. Rosendhal.
\newblock Niho type cross-correlation functions via Dickson polynomials and Kloosterman sums.
\newblock {\em IEEE Transactions on Information Theory}, 52(2): 613 -- 627, 2006.

\bibitem{HELLESETH2008}
T. Helleseth, A. Kholosha, G.J. Ness.
\newblock Characterization of $m$-sequences of lengths $2^{2^{2k}-1}$ and $2^k-1$ with three-valued crosscorrelation.
\newblock  {\em IEEE Transactions on Information Theory}, 53(6): 2236 -- 2245, 2007.

\bibitem{HELLESETH2008}
T. Helleseth, A. Kholosha.
\newblock On the equation $x^{2^l+1}+x+a=0$ over {$GF(2k)$}.
\newblock {\em Finite Fields and Their Applications}, 14(1):159 -- 176, 2008.

\bibitem{HELLESETH2010}
T. Helleseth, A. Kholosha.
\newblock $x^{2^l+1}+x+a$ and related affine polynomials over {$GF(2k)$}.
\newblock {\em Cryptography and Communications}, 2(1):85 -- 109, 2010.


\end{thebibliography}

\end{document}